\newif\iflncs
\lncsfalse 

\iflncs
	\documentclass[envcountsect,envcountsame,orivec,runningheads]{llncs}
\else
	\documentclass[11pt,letterpaper]{article}
\fi

\iflncs
	
\else
	\usepackage{fullpage} 
\fi

\usepackage{amsmath}
\usepackage{amssymb}
\iflncs

\else
	\usepackage{amsthm}
\fi

\usepackage[colorlinks=true,linkcolor=black,citecolor=black,urlcolor=red]{hyperref}
\usepackage{verbatim} 

\begin{document}

\iflncs

\else
	\theoremstyle{plain}
	\newtheorem{theorem}{Theorem}[section]
	\newtheorem{corollary}{Corollary}[section]
	\newtheorem{lemma}{Lemma}[section]
	\theoremstyle{definition}
	\newtheorem{definition}[theorem]{Definition}
	\newtheorem{claim}{Claim}[section]
	\newtheorem{proposition}{Proposition}[section]

\fi

\renewcommand{\proofname}{{\bf Proof}}


\iflncs
  \renewcommand\paragraph[1]{\medskip \noindent {\bf #1}} 
\fi

\newenvironment{proofof}[1]{\smallskip
\noindent {\bf Proof of #1.  }}{\hfill$\Box$
\smallskip}

\newenvironment{reminder}[1]{\smallskip
\noindent {\bf Reminder of #1  }\em}{
\smallskip}

\def\tsumprob{\textsc{-}\textsc{TargetSum}}

\def\mld{\textsc{-}\textsc{MinLinDependence}}
\def\ldprob{\textsc{-LinDependence}}
\def\Z{\mathbb{Z}}

\def\mathspan{\mathsf{span}}

\def\deq{\mathrel{:=}}
\def \iouseful {\text{useful}}
\def \BP {{\sf BP}}
\def \coRP {{\sf coRP}}
\def \EXACT {{\sf EXACT}}
\def \SYM {{\sf SYM}}
\def \SAC {{\sf SAC}}
\def \SUBEXP {{\sf SUBEXP}}
\def \ZPSUBEXP {{\sf ZPSUBEXP}}
\def \SYMACC {{\sf SYM\text{-}ACC} }
\def \QED {{\hfill$\Box$}}
\def \PH {{\sf PH}}
\def \RP {{\sf RP}}
\def \coMA {{\sf coMA}}
\def \ZPTISP {{\sf ZPTISP}}
\def \REXP {{\sf REXP}}
\def \coNP {{\sf coNP}}
\def \BPP {{\sf BPP}}
\def \NC {{\sf NC}}
\def \ZPE {{\sf ZPE}}
\def \NE {{\sf NE}}
\def \E {{\sf E}}
\def \poly { \text{\rm poly} }
\def \TC {{\sf TC}}
\def \DTS {{\sf DTS}}
\def \R {{\cal R}}
\def \Z {{\mathbb Z}}
\def \P {{\sf P}}
\def \MA {{\sf MA}}
\def \AM {{\sf AM}}
\def \MATIME {{\sf MATIME}}
\def \QP {{\sf QP}}
\def \coNQP {{\sf coNQP}}
\def \NP {{\sf NP}}
\def \EXP {{\sf EXP}}
\def \NTISP {{\sf NTISP}}
\def \DTISP {{\sf DTISP}}
\def \TISP {{\sf TISP}}
\def \T {{\sf TIME}}
\def \TIME {\T}
\def \Sig[#1] {{\sf \Sigma}_{#1} }
\def \Pie[#1] {{\sf \Pi}_{#1} }
\def \NTS {{\sf NTS}}
\def \NSP {{\sf NSPACE}}
\def \NSPACE {{\sf NSPACE}}
\def \BPACC {{\sf BPACC}}
\def \ACC {{\sf ACC}}
\def \ASP {{\sf ASPACE}}
\def \io {\textrm{\it io-}}
\def \ro {\textrm{\it ro-}}
\def \ATISP {{\sf ATISP}}
\def \SIZE {{\sf SIZE}}
\def \AT {{\sf ATIME}}
\def \AC {{\sf AC}}
\def \BPTIME {{\sf BPTIME}}
\def \SPACE {{\sf SPACE}}
\def \RE {{\sf RE}}
\def \ZPSUBEXP {{\sf ZPSUBEXP}}
\def \coREXP {{\sf coREXP}}
\def \NQL {{\sf NQL}}
\def \QL {{\sf QL}}
\def \RTIME {{\sf RTIME}}
\def \NSUBEXP {{\sf NSUBEXP}}
\def \MAEXP {{\sf MAEXP}}
\def \PP {{\sf PP}}
\def \PSPACE {{\sf PSPACE}}
\def \NT {{\sf NTIME}}

\def \NTIME {\NT}

\def \ATIME {\AT}

\def \SAT {\mathsf{SAT}}

\def \NTIBI {{\sf NTIBI}}

\def \TIBI {{\sf TIBI}}

\def \QBFk {{\text{QBF}_k}}

\def \coNT {{\sf coNTIME}}

\def \coNTISP {{\sf coNTISP}}

\def \coNTIME {\coNT}

\def \coNTIBI {{\sf coNTIBI}}
\def \MOD {{\sf MOD}}
\def \ZPEXP {{\sf ZPEXP}}
\def \ZPTIME {{\sf ZPTIME}}
\def \ZPP {{\sf ZPP}}
\def \DT {{\sf DTIME}}
\def \coNE {{\sf coNE}}
\def \DTIME {\DT}

\def \isin {\subseteq}

\def \isnotin {\nsubseteq}

\def \L {{\sf LOGSPACE}}

\def \K {{\mathbb K}}
\def \F {{\mathbb F}}

\def \LOGSPACE {\L}

\def \N {{\mathbb N}}
\def \NQP {{\sf NQP}}
\def \NEXP {{\sf NEXP}}

\def \coNEXP {{\sf coNEXP}}

\def \SIZE {{\sf SIZE}}

\def \eps {\varepsilon}

\def \kSUMq {{k\mathrm{SUM}_q}}

\def \kVectorSUM{k\text{-Vector-Sum}}

\def\pr{\mathrm{Pr}}

\def \minksum {{\text{Closest-}}\kSUM}
\newcommand{\cftwo}[1]{\left\lceil \frac{#1}{2} \right\rceil}

\newcommand{\ints}[2]{[#1,#2]}

\def\ksumfree{\textrm{$k$-sum-free}}

\newcommand{\Kevin}[1]
{\begin{center} \framebox{ \parbox{ 15cm }
{\textcolor[rgb]{0.1,0.3,0.7}{{\bf Kevin:} #1}}} \end{center}}

\newcommand{\Amir}[1]
{\begin{center} \framebox{ \parbox{ 15cm }
{\textcolor[rgb]{0.5, 0, 1}{{\bf Amir:} #1}}} \end{center}}

\newcommand{\aanote}[1]
{\begin{center} \framebox{ \parbox{ 15cm }
{\textcolor[rgb]{1, 0, 0}{{\bf Amir:} #1}}} \end{center}}

\def\th{{\text{\rm th}}}

\newcommand{\ksumname}[3]{{(#1,#3)}\textrm{-}\mathtt{SUM}}
\newcommand{\ksumnew}[2]{{(#1,#2)}\textrm{-}\mathtt{SUM}}
\newcommand{\kvecsumname}[4]{{(#1,#3)}
	\textrm{-}\mathtt{VECTOR}\textrm{-}\mathtt{SUM}}
\newcommand{\kcliquename}[3]{{(#1,#2,#3)}\textrm{-}\mathtt{CLIQUE}}
\newcommand{\ewkcliquename}[4]{{(#1,#4)}
	\textrm{-}\mathtt{EW}\textrm{-}\mathtt{CLIQUE}}
\newcommand{\ewkcliquenameshort}[1]{{#1}
	\textrm{-}\mathtt{EW}\textrm{-}\mathtt{CLIQUE}}
\newcommand{\nwkcliquename}[4]{{(#1,#4)}
  \textrm{-}\mathtt{NW}\textrm{-}\allowbreak\mathtt{CLIQUE}}
\newcommand{\nwkcliquenameshort}[1]{{#1}
	\textrm{-}\mathtt{NW}\textrm{-}\mathtt{CLIQUE}}
\newcommand{\nwkDSname}[3]{{(#1,#2,#3,)}\textrm{-}\mathtt{DOMINATING-SET}} 

\def\kClique{k\text{-Clique}}
\def\kSUM{k\text{-SUM}}
\def\kSUMZ{k\text{-SUM-}{\mathbb Z}}
\def\W{{\sf W}}
\def\FPT {{\sf FPT}}
\def\subsetsum{\textrm{Subset-SUM}}
\def\eewc{\textrm{Exact Edge-Weight Clique}}
\def\maxclique{\textrm{Max-Clique}}
\def\kDS{k\text{-}\allowbreak\text{Dominating Set}}
\def\patrascu{\text{P{\v a}tra{\c s}cu}}

\def\calD{\mathcal{D}}

\newcommand{\evec}[2]{#1[#2]} 
\newcommand{\ivec}[2]{#1_{#2}} 

\newcommand{\ievec}[3]{\evec{\ivec{#1}{#2}}{#3}}

\newcommand{\vid}[2]{\evec{\boldv_{#1}}{#2}}
\newcommand{\xid}[2]{\evec{\boldv_{x_{#1}}}{#2}}
\newcommand{\xind}[1]{{x_{#1}}}
\newcommand{\boldf}{\ensuremath{\mathbf{f}}}
\newcommand{\boldx}{\ensuremath{\mathbf{x}}}
\newcommand{\boldr}{\ensuremath{\mathbf{r}}}
\newcommand{\boldw}{\ensuremath{\mathbf{w}}}
\newcommand{\boldt}{\ensuremath{\mathbf{t}}}
\newcommand{\boldu}{\ensuremath{\mathbf{u}}}
\newcommand{\boldv}{\ensuremath{\mathbf{v}}}
\newcommand{\boldy}{\ensuremath{\mathbf{y}}}
\newcommand{\boldz}{\ensuremath{\mathbf{z}}}
\newcommand{\vx}[1]{\boldv({#1})}
\newcommand{\ux}[1]{\boldu({#1})}
\def\zerov{\boldzero}

\newcommand{\vsubx}[1]{\boldv_{#1}}

\def\boldgamma{\ensuremath{\boldsymbol{\gamma}}}
\def\boldtau{\ensuremath{\boldsymbol{\tau}}}

\def\carryvec{\boldgamma}

\def\mappingchar{f}
\newcommand{\mapping}[3]{\mappingchar_{#2}(#1,#3)}

\def\SmallkSUM{\text{Small-}k\text{-SUM}}

\def\boldzero{\ensuremath{\boldsymbol{0}}}
\newcommand{\indvec}[2]{\ensuremath{#2 \cdot e_{#1}}}
\def\func{\boldsymbol{\eta}}
\newcommand{\funci}[1]{\func_{f_{#1}}}

\iflncs

\else
\title{Losing Weight by Gaining Edges}
\author{Amir Abboud\thanks{Computer Science Department, Stanford 
University}\\{\tt abboud@cs.stanford.edu} \and Kevin Lewi\footnotemark[1]\\{\tt 
klewi@cs.stanford.edu} \and Ryan Williams\footnotemark[1]
\footnote{Supported in part by a David Morgenthaler II Faculty Fellowship, and 
NSF CCF-1212372.}\\{\tt rrw@cs.stanford.edu}}
\date{}
\fi

\iflncs
\title{Losing Weight by Gaining Edges}
\author{
Amir Abboud \and
Kevin Lewi\and
Ryan Williams}
\institute{Stanford University}
\fi

\maketitle

\begin{abstract}

We present a new way to encode weighted sums into unweighted pairwise 
constraints, obtaining the following results.

\begin{itemize}

\item Define the $\kSUM$ problem to be: given $n$ integers in $[-n^{2k},n^{2k}]$ are there $k$ which sum to zero? 
(It is well known that the same problem over \emph{arbitrary} integers is equivalent to the above definition, by linear-time randomized reductions.) We prove that this definition of $\kSUM$ remains $\W[1]$-hard, and is in fact $\W[1]$-\emph{complete}: $\kSUM$ can be reduced to $f(k)\cdot n^{o(1)}$ instances of $\kClique$.

\item The maximum node-weighted $k$-Clique and node-weighted $k$-dominating set problems can be reduced to $n^{o(1)}$ instances of the \emph{unweighted} $k$-Clique and $k$-dominating set problems, respectively. This implies a strong \emph{equivalence} between the time complexities of the node weighted problems and the unweighted problems: any polynomial improvement on one would imply an improvement for the other.

\item A triangle of weight $0$ in a node weighted graph with $m$ edges can be deterministically found in $m^{1.41}$ time. 
\end{itemize}

%
\end{abstract}


\thispagestyle{empty}

\setcounter{page}{1}

\section{Introduction}


One of the most basic problems over integers, studied in geometry, cryptography, and combinatorics, is $k$-SUM, the parameterized version of the classical NP-complete problem SUBSET-SUM.

\begin{definition}[$\kSUM$] The $\ksumnew{k}{M}$ problem is to determine, given 
  $n$ integers $x_1,\ldots,x_n \in [0,M]$ and a target integer $t\in [0,M]$, if 
  there exists a subset $S \subseteq [n]$ of size $|S|=k$ such that $\sum_{i\in 
  S} x_i = t$.\footnote{Without loss of generality, the range of integers can be 
  $[-M,M]$ and the target integer can be zero.}
	 We define $\kSUM \triangleq 
	\ksumnew{k}{n^{2k}}$.

\end{definition}

Our definition of $\kSUM$ is justified via the following known proposition:

\begin{proposition} \label{random-prime-prop}
	Every instance $S$ of $\ksumnew{k}{M}$ can be randomly reduced in $O(k n \log 
	M)$ time to an instance $S'$ of $\kSUM$ as defined above.
\end{proposition}

That is, there is an efficient randomized reduction from $\kSUM$ over arbitrary 
integers, which we call $\kSUMZ$, to our definition of $\kSUM$ 
(Appendix~\ref{reducing}).
Furthermore, we show in Appendix~\ref{reducing}
that this reduction can be made deterministic under standard hardness 
assumptions.

%

A classical meet-in-the-middle algorithm solves $\kSUM$ in $\tilde{O}(n^{\lceil k/2 \rceil})$ time and it has been a longstanding open problem to obtain an $O(n^{\lceil k/2 \rceil-\eps})$ algorithm for any integer $k\geq 3$ and constant $\eps>0$.
Logarithmic improvements are known for the $k=3$ case~\cite{BDP08,GP14} (that is, the famous $3$-SUM problem). 
The $k$-SUM conjecture~\cite{Patrascu10,AL13} states that $k$-SUM requires $n^{\lceil{k/2}\rceil-o(1)}$ time and is known to imply tight lower bounds for many problems in computational geometry~\cite{GO95,Barrera96,BH01} (and many more) and has recently been used to show conditional lower bounds for discrete problems as well~\cite{Patrascu10,VW09,JV13,AV14}.
A matching $\Omega(n^{\lceil k/2 \rceil})$ lower bound for $\kSUM$ was shown for a restricted model of computation called $k$-linear decision trees (LDTs)~\cite{Erickson95,AC05}, although it was recently shown that depth $O(n^{k/2}\sqrt{\log{n}})$ suffices for $(2k-2)$-LDTs~\cite{GP14}.
It is also known that if there is an unbounded function $s:\mathbb{N} \rightarrow \mathbb{N}$ such that for infinitely many $k$, $\kSUM$ is in $n^{k/s(k)}$, then the Exponential Time Hypothesis is false~\cite{PW10}.

Despite intensive research on this simple problem, our understanding is still lacking in many ways, one of which is from the viewpoint of parameterized complexity. In their seminal work on parameterized 
intractability, Downey and Fellows~\cite{DF92,DF95} proved that $\kSUMZ$ is $\W[1]$-hard and is contained in $\W[P]$. 
The even simpler Perfect Code problem was conjectured to lie between the classes $\W[1]$ and $\W[2]$~\cite{DF95} until Cesati proved it was $\W[1]$-complete in 2002~\cite{Ces02}.
Classifying $\kSUMZ$ within a finite level of the $\W$-hierarchy was open until in 2007, when Buss and 
Islam~\cite{BI07} proved that $\kSUMZ \in \W[3]$.


The primary contribution of this work is a novel and generic way to efficiently 
convert problems concerning sums of numbers into problems on unweighted pairwise 
constraints. We call this technique ``Losing weight by gaining edges" and report 
several interesting applications of it. One application is a new parameterized 
reduction from $\kSUM$ to $\kClique$ and therefore the resolution of the 
parameterized complexity of $\kSUM$ (for numbers in $[-n^{2k},n^{2k}]$). Under 
standard lower bound hypotheses, we also obtain a deterministic reduction from 
$\kSUMZ$ to $\kClique$ as well.

\begin{theorem}
\label{thm:w1}
$\kSUM$ is $\W[1]$-complete.
\end{theorem}


The significance of showing $\W[1]$-hardness for a problem is well known (as it rules out FPT algorithms). 
The significance of showing that a problem is in $\W[1]$ is less obvious, so let us provide some motivation.
First, although 
$\W[1]$-complete problems are probably not FPT, prominent problems in $\W[1]$ 
(such as $\kClique$) can still be solved substantially faster than exhaustive 
search over all ${n \choose k}$ subsets~\cite{NP85}. In contrast, analogous 
problems in $\W[2]$ (such as $k$-Dominating Set) do not have such algorithms 
unless CNF Satisfiability is in $2^{\delta n}$ time for some $\delta < 
1$~\cite{PW10}, which is a major open problem in exact algorithms. Therefore, 
understanding which parameterized problems lie in $\W[1]$ is closely related to 
understanding which problems can be solved faster than exhaustive search. 
Second, showing that a 
problem is in $\W[1]$ rather than $\W[3]$ means that it can be expressed in an 
apparently weaker logic than before, with fewer quantifiers~\cite{FG06}. That 
is, putting a problem in $\W[1]$ decreases the descriptive complexity of the 
problem.


Theorem~\ref{thm:w1} has applications to parameterized complexity, yielding a 
new characterization of the class $\W[1]$ as the problems FPT-reducible to 
$\kSUM$. Since $\kSUM$ is quite different in nature from the previously known 
$\W[1]$-complete problems, we are able to put other such ``intermediate" 
problems in $\W[1]$, including weighted graph problems and problems with 
application to coding theory such as Weight Distribution~\cite{CP12}.

To show that $\kSUM \in \W[1]$, we prove a very tight reduction from $\kSUM$ to $\kClique$.
Given an instance of $\kSUM$ on $n$ numbers, we generate $f(k)\cdot n^{o(1)}$ instances of $\kClique$ on $n$ node graphs, such that one of these graphs contains a $k$-clique if and only if our $\kSUM$ instance has a solution.
This implies that any algorithm for $\kClique$ running in time $O(n^c)$ for some $c\geq2$ yields an algorithm for $\kSUM$ running in time $n^{c+o(1)}$. 
Hence, the $k$-SUM conjecture implies an $n^{\lceil k/2 \rceil - o(1)}$ lower bound for $\kClique$ as well.

Generalizing our ideas further, we are able to prove surprising consequences regarding other weighted problems.

\paragraph{Removing node weights.} Two fundamental graph problems are $\kClique$ and $\kDS$. Natural extensions of these problems allow the input graph to have weights on its nodes. The problem can then be to find a $k$-clique or a $k$-dominating set of minimum or maximum sum of node weights (the {\em min} and {\em max} versions), or to find a $k$-clique or a $k$-dominating set with total weight exactly $0$ (the {\em exact} version, defined below). 

\begin{definition}[The Node-Weight $k$-Clique-Sum Problem]
	\label{def3}
	For integers $k,M > 0$, the 
	$\nwkcliquename{k}{n}{m}{M}$ problem is to determine, given a graph $G$, a node-weight 
	function $w : V(G) \to \ints{0}{M}$, and a target weight $t \in \ints{0}{M}$, 
	if there is a set $S$ of $k$ nodes which form a clique such that $\sum_{v \in 
	S} w(v) = t$. We define the Node-Weight $k$-Clique-Sum problem as $
	\nwkcliquename{k}{n}{m}{n^{2k}}$.
\end{definition}

\begin{definition}[Node-Weight $k$-Dominating-Set-Sum]
For an integer $k> 0$, the Node-Weight $k$-Dominating-Set-Sum problem is to determine, given a graph $G$, a node-weight 
	function $w : V(G) \to \ints{0}{n^{2k}}$, and a target weight $t \in \ints{0}{n^{2k}}$, 
	if there is a set $S$ of $k$ nodes which form a dominating set such that $\sum_{v \in 
	S} w(v) = t$.  
\end{definition}

These additional node weights increase the expressibility of the problem and allow us to capture more applications. How much harder are these node weighted versions compared to the unweighted versions?
By weight scaling arguments, one can show that the ``exact" version is harder than the max and min versions, in the sense that any algorithm for ``exact" implies an algorithm for max or min with only a logarithmic overhead \cite{NLZ12} (and Theorem 3.3 in \cite{VW09}). But how much harder is (for example) Node-Weight $k$-Clique-Sum than the case where there are no weights at all?

For $k$ divisible by $3$, the best $\kClique$ algorithms reduce the problem to $3$-Clique on $n^{k/3}$ nodes, then use an $O(n^{\omega})$ time algorithm for triangle detection~\cite{IR77} for a running time of $O(n^{\omega k /3})$~\cite{NP85}.
This reduction to the $k=3$ case works for the node weighted case as well; combined with the recent $n^{\omega+o(1)}$ algorithms for node weighted triangle ($3$-clique) problems \cite{CL09,VW09}, we obtain $n^{\omega k /3+o(1)}$ running times for node weighted $k$-clique problems.
The best $\kDS$ algorithms reduce the problem to a rectangular matrix multiplication of matrices of dimensions $n^{k/2}\times n$ and $n \times n^{k/2}$ and run in time $n^{k+o(1)}$~\cite{EG04}. These algorithms allow us to find all $k$-dominating sets in the graph and therefore can also solve the node weighted versions without extra cost.

Therefore, the state of the art algorithms for $\kClique$ and $\kDS$ suggest that adding node weights does not make the problems much harder. Is that due to our current algorithms for the unweighted problems, or is there a deeper connection? 
Using the ``Losing weight by gaining edges" ideas, we show that the node weighted versions of $\kClique$ and $\kDS$ (and, in fact, {\em any} problem that allows us to implement certain ``pairwise constraints") are essentially ``equivalent" to the unweighted versions.

\begin{theorem}
\label{thm:NW}
If $\kClique$ on $n$ node and $m$ edge graphs can be solved in time $T(n,m,k)$, 
then Node-Weight $k$-Clique-Sum on $n$ node and $m$ edge graphs can be solved in 
time $n^{o(1)} \cdot T(kn, k^2 m,k)$.
If $\kDS$ on $n$ node graphs can be solved in time $T(n,k)$, then Node-Weight 
$k$-Dominating-Set on $n$ node graphs can be solved in time $n^{o(1)} \cdot 
T(k^2n,k)$.
\end{theorem}

Interestingly, Theorem~\ref{thm:NW} yields a short and simple $n^{\omega+o(1)}$ algorithm for the node-weighted triangle problems, while a series of papers were required to recently conclude the same upper bound using different techniques \cite{VW06,VWY06,CL09,VW09}. Moreover, unlike the previous techniques, our approach extends to $k>3$ and applies to more problems like $\kDS$.

Applying the result of Theorem~\ref{thm:NW} to the 
$O(m^{\frac{2\omega}{\omega+1}})$ triangle detection algorithm of Alon, Yuster 
and Zwick~\cite{AYZ97}, we obtain a \emph{deterministic} algorithm for 
Node-Weight Triangle-Sum in sparse graphs, improving the previous 
$n^{\omega+o(1)}$ upper bound~\cite{VW09} and matching the running time of the 
best randomized algorithm~\cite{VW09}.

\begin{corollary}
Node-Weight Triangle-Sum can be solved deterministically in $m^{1.41+o(1)}$ time.
\end{corollary}

\subsection{Overview of the Proofs}

Let us give some intuition for Theorem~\ref{thm:w1}. Both the containment 
in $\W[1]$ and the hardness for $\W[1]$ require new technical ideas. Downey and 
Fellows~\cite{DF92,DF95} proved that $\kSUMZ$ is $\W[1]$-hard by a reduction 
requiring fairly large numbers: they are exponential in $n$, but can still be 
generated in an FPT way. To prove that $\kSUM$ is $\W[1]$-hard even when the 
numbers are only exponential in $k \log n$, we need a much more efficient encoding of $k$-Clique instances. We apply some machinery from 
additive combinatorics, namely a construction of large sets of integers avoiding 
trivial solutions to the linear equation $\sum_{i=1}^{k-1} x_i = 
(k-1)x_k$~\cite{OBryant11}. 
These sets allow us to efficiently ``pack'' a $k$-Clique instance into a $({k 
\choose 2}+k)$-SUM instance on small numbers.

Proving that $\kSUM$ is in $\W[1]$ takes several technical steps. We 
provide a parameterized reduction from $\kSUM$ on $n$ numbers to only $f(k)\cdot n^{o(1)}$ graphs on $O(kn)$ nodes, such that some graph has a $k$-clique if 
and only if the original $n$ numbers have a $\kSUM$. To efficiently reduce from 
numbers to graphs, we first reduce the numbers to an analogous problem on 
vectors. We define an intermediate problem $\kvecsumname{k}{n}{M}{d}$, in which 
one is given a list of $n$ vectors from $\{-kM,\ldots,0,\ldots,kM\}^d$, and is 
asked to determine if there are $k$ vectors which sum to the all-zero vector. We 
give an FPT reduction from $\kSUM$ to $\kvecsumname{k}{n}{M}{d}$ where $M$ and 
$d$ are ``small'' (such that $M^d$ is approximately equal to the original 
weights of the $\kSUM$ instance). Next, we ``push'' the weights in these vectors 
onto the edges of a graph connecting the vectors, where the edge weights are 
much smaller than the original numbers: we reduce from 
$\kvecsumname{k}{n}{M}{d}$ to edge-weighted $k$-clique-sum using a 
polynomial ``squaring trick'' which creates a graph with ``small" edge weights, 
closely related in size to $M$. Finally, we reduce from the weighted problem to 
the unweighted version of the problem by brute-forcing all feasible weight 
combinations on the edges; as the edge weights are small, this creates 
$f(k)\cdot n^{o(1)}$ unweighted $k$-Clique instances for some function $f$.


Combining all these steps into one, one can view our approach as follows.
We enumerate over all ${k \choose 2}$-tuples of numbers $t=(\alpha_{i,j})_{i,j \in [k]}$ such that $\sum_{i,j} \alpha_{i,j} = 0$ where $\alpha_{i,j} \in [-M,M]$ for $M=f(k) \cdot \poly\log{n}$, and for each such tuple $t$ we generate an instance of the unweighted problem. In this instance, two nodes are allowed to both be a part of our final solution (e.g. there is an edge between them in the $k$-clique case) if and only if some expression on the weights of the objects $v_i$ and $v_j$ evaluates to $F(w(v_i),w(v_j))=\alpha_{i,j}$.
The formulas are defined, via the ``squaring trick", in such a way that there are $k$ nodes satisfying these ${k \choose 2}$  equations for some ${k \choose 2}$-tuple $t$ if and only if the sum of the weights of these $k$ nodes is $0$.

To implement our approach for $\kDS$ we follow similar steps, except that we cannot implement the constraints on having a certain pair of objects in our solution by removing the edge between them anymore, since this does not prevent them from being in a feasible $k$-dominating set.
This can be done, however, by adding extra nodes $X$ to the graph such that the inclusion of pairs of nodes $v_i,v_j$ in the solution $S$ that do not satisfy our equations, $F(w(v_i),w(v_j)) \neq \alpha_{i,j}$, will prevent $S$ from dominating all the nodes in $X$.

\subsection{Related Work}

 There has been recent 
work in relating the complexity of $\kSUM$ and variations of $\kClique$ for the 
specific case of $k=3$. $\patrascu$~\cite{Patrascu10} shows a tight reduction from 
$3$-SUM to listing $3$-cliques; a reduction from listing $3$-cliques to $3$-SUM 
is given by Jafargholi and Viola~\cite{JV13}. Vassilevska and 
Williams~\cite{VW09} consider the exact edge-weight $3$-clique problem and give 
a tight reduction from $3$-SUM. For the case of $k>3$, less is known, as the 
techniques used for the case of $k=3$ do not seem to generalize easily. Abboud 
and Lewi~\cite{AL13} give reductions between $\kSUM$ and various exact 
edge-weighted graph problems where the goal is to find an instance of a specific 
subgraph whose edge weights sum to $0$.

\section{Preliminaries} \label{prelims}

For $i < j \in \Z$, define $\ints{i}{j} \triangleq \{ i, \ldots, j\}$. As 
shorthand, we define $[n] \triangleq \ints{1}{n}$. For a vector $\boldv \in 
\mathbb{Z}^d$, we denote by $\evec{\boldv}{j}$ the value in the $j^{\th}$ 
coordinate of $\boldv$. We let $\boldzero$ denote the all zeros vector. The 
default domain and range of a function is $\N$.



\medskip


We define the $\kClique$ problem as follows.

\begin{definition}[The $k$-Clique Problem]
	For integers $k,n,m > 0$, the $k$-clique problem is to determine, given a 
	graph $G$, if there is a size-$k$ subset $S \subseteq [n]$ such that $S$ is a 
	clique in $G$.
\end{definition}

The following problems are referred to in Corollary~\ref{cor:subsetsum}. They 
are simply the unparameterized versions of $\kSUM$ and Exact Edge-Weight 
$\kClique$, respectively.

\begin{definition}[The $\subsetsum$ Problem] The $\subsetsum$ problem is to 
	determine, given a set of integers $x_1,\ldots,x_n, t$, if there exists a 
	subset $S \subseteq [n]$ such that $\sum_{i \in S} x_i = t$.
\end{definition}

\begin{definition}[The Exact Edge-Weight Clique Problem]
	For integers $n,m,M > 0$, the Exact Edge-Weight $k$-Clique problem is to 
	determine, given an instance of a graph $G$ on $n$ vertices and $m$ edges, a 
	weight function $w : E(G) \to \ints{-M}{M}$, if there exists a set of nodes 
	which form a clique with total weight $0$.
\end{definition}

\section{From Numbers to Edges} \label{sec3}

Our results begin by showing how to reduce $\kSUM$ to $\kClique$. To do this, we 
first give a new reduction from $\kSUM$ to $\kVectorSUM$ on $n$ vectors in $C^d$ 
for a set $C$ which is relatively small compared to the numbers in the original 
instance. From $\kVectorSUM$, we give a reduction to Edge-Weight $k$-Clique-Sum 
with small weights.
Then, we can brute-force all possibilities for the $\binom{k}{2}$ edge weights 
for $\kSUM$ and reduce to the (unweighted) $\kClique$ problem. Altogether, we 
conclude that $\kSUM$ is in $\W[1]$.


\subsection{Reducing $\kSUM$ to $\kVectorSUM$}

We present a generic way to map numbers into vectors over small numbers such 
that the $k$-sums are preserved. We define the $k$-Vector-Sum problem as 
follows.

\begin{definition}[The $k$-Vector-Sum Problem]
	\label{kvecsum-def}
	For integers $k,n,M,d > 0$, the $k$-vector-sum problem 
	$\kvecsumname{k}{n}{M}{d}$ is to determine, given vectors $\boldv_1, \ldots, 
	\boldv_n, \boldt \in [0,kM]^d$, if there is a size-$k$ subset $S \subseteq 
	[n]$ such that $\sum_{i \in S} \boldv_i = \boldt$.
\end{definition}

Note that the problem was considered by Bhattacharyya et al.~\cite{BIWX11} and 
also by Cattaneo and Perdrix~\cite{CP12}.


\begin{lemma} \label{lem1}
	Let $k,p,d,s,M \in \N$ satisfy $k<p$, $p^d \geq  k M + 1$, and 
	$s=(k+1)^{d-1}$. There is a collection of mappings $\mappingchar_1, \ldots, 
	\mappingchar_s : [0,M] \times [0,kM] \to [-kp,kp]^d$, each computable in time 
	$O(\poly\log M + k^d)$, such that for all numbers $x_1,\ldots,x_k \in [0,M]$ 
	and targets $t\in [0,kM]$, \[ \sum_{j=1}^k x_j = t \qquad \Leftrightarrow 
	\qquad \exists\,i \in [s] \quad \text{such that} \quad \sum_{j=1}^k 
	\mapping{x_j}{i}{t} = \zerov. \]
\end{lemma}


\begin{proof}
	For $x \in [0,kM]$, define $\vsubx{x} \in [0,p-1]^d$ to be the representation 
	of $x$ in base $p$---that is, let $\evec{\vsubx{x}}{1}, \ldots, 
	\evec{\vsubx{x}}{d} \in [0,p-1]$ be unique integers such that $x = 
	\evec{\vsubx{x}}{1} \cdot p^0 + \cdots + \evec{\vsubx{x}}{d} \cdot p^{d-1}$. 
	Note that $p^d > kM$.
	
	
	For every $(d-1)$-tuple of numbers $\gamma = (c_1,\ldots,c_{d-1})\in 
	[0,k]^{d-1}$ and a number $t \in [0,kM]$, define the vector $\boldtau_\gamma$ 
	to be such that \[
	\evec{\boldtau_\gamma}{j} = \begin{cases}
		\evec{\vsubx{t}}{1} + c_1 \cdot p, &\text{if $j=1$} \\
		\evec{\vsubx{t}}{j} - c_{j-1} + c_{j} \cdot p, &\text{if $2\leq j \leq d-1$} 
		\\
		\evec{\vsubx{t}}{d} - c_{d-1} , &\text{if $j=d$} \\
		\end{cases}
	\]
	Intuitively, each tuple $\gamma$ represents a possible choice of the 
	``carries'' obtained in each component when computing the sum of their 
	corresponding base-$p$ numbers, and $\boldtau_\gamma$ corresponds to the 
	target vector in this base-$p$ representation.
	Note that there are at most $s = (k+1)^{d-1}$ such vectors. We arbitrarily 
	number these vectors as $\boldt_1,\ldots,\boldt_s$, and for all $i \in [s]$ we 
	define the mapping \[ \mapping{x}{i}{t} = k\cdot \vsubx{x} -\boldt_i. \]

	Consider a set of $k$ numbers $S = \{ x_1,\ldots,x_k \} \subseteq [0,M]$. We 
	claim that $\sum_{j=1}^k x_j = t$ if and only if there exists some $i \in [s]$ 
	such that $\sum_{j=1}^k \mapping{x_j}{i}{t} = \zerov$. First observe that for 
	any vectors $\boldu_1, \ldots, \boldu_k,\boldt \in \Z^d$, $\sum_{j=1}^k 
	(k\cdot \boldu_j - \boldt ) = \zerov$ if and only if $\sum_{j=1}^k \boldu_j = 
	\boldt$. It suffices to show that $\sum_{j=1}^k x_j = t$ if and only if there 
	exists some $\gamma \in[0,k]^{d-1}$ such that $\sum_{j=1}^k \vsubx{x_j} = 
	\boldtau_\gamma$.
		
	For the first direction, $t = \sum_{i=1}^k x_i$. We can write this equation in 
	base-$p$ representation as
	\[
	\sum_{j=0}^{d-1} \evec{\vsubx{t}}{j+1} \cdot p^j = \sum_{i=1}^{k} \left( 
	\sum_{j=0}^{d-1} \xid{i}{j+1} \cdot p^j \right) = \sum_{j=0}^{d-1} \left( 
	\sum_{i=1}^{k} \xid{i}{j+1} \right) \cdot p^j.\]
	Therefore, $\sum_{i=1}^{k} \xid{i}{1} \in [0,k(p-1)] = \evec{\vsubx{t}}{1} + 
	c_1\cdot p$ for some $c_1 \in [0,k]$, $c_1
	+ \sum_{i=1}^{k} \xid{i}{2} = c_2 \cdot p + \evec{\vsubx{t}}{2}$ for some $c_2 
	\in [0,k]$, and so on. Finally, $c_{d-1} + \sum_{i=1}^{k} \xid{i}{d} = 
	\evec{\vsubx{t}}{d}$. Letting $\gamma = (c_1,\ldots,c_{d-1}) \in [0,k]^{d-1}$, 
	these equations imply exactly that $\sum_{i=1}^k \vsubx{x_i} = 
	\boldtau_\gamma$.
	
	For the other direction, suppose $\sum_{i=1}^k \vsubx{x_i} = \boldtau_\gamma$ 
	where $\gamma=(c_1,\ldots,c_{d-1}) \in [0,k]^{d-1}$. Again, converting from 
	base-$p$ representation we obtain $\sum_{i=1}^k x_i = \sum_{j=0}^{d-1} \left( 
	\sum_{i=1}^{k} \xid{i}{j+1} \right) \cdot p^j$, and using the definition of 
	$\boldtau_\gamma$ we can write the sum of the variables in $S$ as
	\begin{align}
		\sum_{i=1}^k x_i
		&= ( \evec{\vsubx{t}}{1} + c_1 \cdot p^1) + \left( \sum_{i=1}^{k} \xid{i}{1} 
		\right) \cdot p^1 + \cdots + \left( \sum_{i=1}^{k} \xid{i}{d} \right) \cdot 
		p^{d-1} \nonumber \\
		&=  \evec{\vsubx{t}}{1}  + ( \evec{\vsubx{t}}{2} + c_2 \cdot p^1) \cdot p^1 
		+ \cdots + \left( \sum_{i=1}^{k} \xid{i}{d} \right) \cdot p^{d-1} \nonumber 
		\\
		&= \cdots = \evec{\vsubx{t}}{1}  + \evec{\vsubx{t}}{2} \cdot p^1+\cdots+ 
		\evec{\vsubx{t}}{d} \cdot p^{d-1} = t, \nonumber
	\end{align}
	which completes the proof.
\end{proof}

\begin{corollary}
	Let $k,p,d,M,n > 0$ be integers with $k<p$ and $p^d \geq k M + 1$. $\kSUM$ on 
	$n$ integers in the range $[0,M]$ can be reduced to $O(k^d)$ instances of 
	$\kvecsumname{k}{n}{p-1}{d}$ on $n$ vectors in $[0,p-1]^d$.
\end{corollary}

\subsection{Reducing to $k$-Clique}

Here, we consider a generalization of the $\kSUM$ problem---namely, the
Node-Weight $\kClique$-Sum problem. We give a reduction from Node-Weight 
$\kClique$-Sum to Edge-Weight $\kClique$-Sum (defined below), where the new edge 
weights are much smaller than the original node weights. We then show how to 
reduce to many instances of the unweighted version of the problem, where each 
instance corresponds to a possible setting of edge weights. Then, we give an 
application of this general reduction to the Node-Weight $\kClique$-Sum problem.

\begin{definition}[The Edge-Weight $k$-Clique-Sum Problem]
	For integers $k,M > 0$, the edge-weight $k$-clique-sum problem 
	$\ewkcliquename{k}{n}{m}{M}$ is to determine, given a graph $G$, an 
	edge-weight function $w : E(G) \to \ints{0}{M}$, and a target weight $t \in 
	\ints{0}{M}$, if there is a set $S$ of $k$ nodes which form a clique such that 
	$\sum_{(u,v) \in S} w(u,v) = t$.
\end{definition}

\begin{lemma} \label{lem:ewclique}
  Let $k,p,d,M > 0$ be integers such that $k<p$ and $p^d \geq k M + 1$, and let 
  $M' = O(k^3dp^2)$. The $\nwkcliquename{k}{n}{m}{M}$ problem can be 
  deterministically reduced to $O(k^d)$ instances of
	$\ewkcliquename{k}{n}{m}{M'}$ in time $O(k^d \cdot n^2 \cdot \poly\log M)$.
\end{lemma}

Thinking of $p+d$ as ``small'', but $\poly(p,d) \approx kM$ as ``large'', we get 
a substantial reduction in the weights of the problem by ``spreading'' the node 
weights over the edges.

\begin{proof}
	Let $G=(V,E)$ be a graph with a node weight function $w : V \rightarrow [0,M]$ 
	and a target number $t \in [0,kM]$. Recall the mappings
	$\mappingchar_i : [0,M] \times \times [0,kM] \to [-kp,kp]^d$ for $i \in [s]$ 
	from Lemma~\ref{lem1}, which maps numbers from $[0,M]$ into a collection of $s 
	= O(k^d)$ length-$d$ vectors with entries in $[-kp, kp]$. We translate the 
	node-weight vector problem into an edge-weight problem via a ``squaring 
	trick,'' as follows. For each $i \in [s]$, we define an edge weight function 
	$w_i : E \rightarrow [-M',M']$.
	For $(u,v)\in E$, let $\boldu = \mapping{w(u)}{i}{t}$ and $\boldv = 
	\mapping{w(v)}{i}{t}$, and define \[w_i(u,v) \triangleq \sum\limits_{j=1}^d 
	\left( \evec{\boldu}{j}^2 + \evec{\boldv}{j}^2 + 2(k-1) \evec{\boldu}{j} \cdot 
	\evec{\boldv}{j} \right).\]
	Note that for $M' = O(k d p^2)$, $w_i(u,v) \in [-M',M']$.
	We show that there is a $k$-clique in $(G,w)$ of node-weight $t$ if and only 
	if for some $i\in[s]$, the edge-weighted graph $(G,w_i)$ contains a $k$-clique 
	of edge-weight $0$. First, observe that for any $k$ vectors $\boldv_1, \ldots, 
	\boldv_k \in \Z^d$,
	\[ \sum_{i=1}^k \boldv_i = \zerov \iff \sum_{j=1}^d \left(\sum_{i=1}^k 
	\vid{i}{j} \right)^2 = 0.\]

	
	Consider a set $S=\{ u_1,\ldots,u_k \}\subseteq V$ that forms a $k$-clique in 
	$G$. For any $i \in [s]$ and $u_a, u_b \in S$, let $\boldu_a = 
	\mapping{w(u_a)}{i}{t}$ and $\boldu_b = \mapping{w(u_b)}{i}{t}$. Then, the 
	edge-weight of $S$ in $(G,w_i)$ is
	\[
 \sum\limits_{1\leq a < b \leq k} w_i(u_{a},u_{b}) = (k-1) \sum\limits_{a=1}^k 
 \sum\limits_{j=1}^d \evec{\boldu_a}{j}^2 + 2 (k-1) \sum_{1\leq a < b \leq k} 
 \sum\limits_{j=1}^d \evec{\boldu_a}{j} \cdot \evec{\boldu_b}{j}. \]

	Since the sum is evaluated over all pairs $a,b \in [k]$ where $a < b$, the 
	above quantity is equal to

	\[ (k-1)\cdot \sum_{j=1}^d \left(\sum_{u\in S} \mapping{w(u)}{i}{t}[j] 
	\right)^2 . \]

	Therefore, for all $i \in [s]$, the edge-weight of $S$ in $(G,w_i)$ equals $0$ 
	if and only if the sum of the vectors $\sum_{u \in S} \mapping{w(u)}{i}{t}$ 
	equals $\zerov$. And, by the properties of the mappings $f_i$ from 
	Lemma~\ref{lem1}, the latter occurs for some $i\in [s]$ if and only if the 
	node-weight of $S$ in $(G,w)$, $\sum_{u \in S} w(u)$, is equal to $t$, as 
	desired.
	\iflncs \qed \fi
\end{proof}

We observe that in the graphs produced by the above reduction, all $k$-cliques 
have non-negative weight. Therefore, Lemma~\ref{lem:ewclique} can also be viewed 
as a reduction to the ``minimum-weight'' $\kClique$ problem with edge weights, 
where the edge sum is minimized.

\medskip

Finally, small weights on edges can simply be eliminated using a brute-force 
step.



\begin{lemma} \label{weight-removal}
	For all integers $k, M > 0$, there is an $O(M^{k \choose 2} \cdot n^2)$ time 
	reduction from the problem $\ewkcliquename{k}{n}{m}{M}$ to $O(M^{{k \choose 
	2}})$ instances of $\kClique$ on $n$ nodes and $m \cdot \binom{k}{2}$ edges.
\end{lemma}

\begin{proof}
Let $G=(V,E)$ be an edge-weighted graph with weight function $w : E
\rightarrow [-M,M]$. For all possible ${k \choose 2}$-tuples $\alpha=
(\alpha_{1,2},\alpha_{1,3},\ldots,\alpha_{k-1,k})$ of integers from $[-M,M]$ 
such
that $\sum_{i < j} \alpha_{i,j} = 0$, we make a $k$-partite graph $G_\alpha$ 
with $|V|$ nodes in each part---that is,  $V(G_\alpha) = \{(v,i) \mid i\in [k], 
v \in V\}$. We put an
edge in $G_\alpha$ between node $(u,i)$ and node $(v,j)$ if and only if
$i< j$, $(u,v) \in E$, and $w(u,v) = \alpha_{i,j}$.
Observe that $O(M^{{k \choose 2}})$ graphs $G_\alpha$ are generated.
We claim that $G$ has a $k$-clique of weight $0$ if and only if at least one of 
the unweighted graphs $G_\alpha$ contains a $k$-clique.

For the first direction, assume that for some $\alpha = 
(\alpha_{1,2},\alpha_{1,3},\ldots,\alpha_{k-1,k})$, the graph $G_\alpha$ 
contains a $k$-clique. Then, by the $k$-partite construction, it must have the 
form $\{(v_1,1),\ldots,(v_k,k)\}\subseteq V(G_\alpha)$. Now let 
$S=\{v_1,\ldots,v_k \} \subseteq V$ and note that by the construction of 
$G_\alpha$, $S$ must form a $k$-clique in $G$ with weight equal to $\sum_{i < j} 
\alpha_{i,j} = 0$.
For the other direction, assume $\{u_1,\ldots,u_k\}\subseteq V$ is a $k$-clique 
in $G$ of weight $0$, then consider the $\binom{k}{2}$-tuple 
$\alpha=(w(u_1,u_2),w(u_1,u_3),\ldots,w(u_{k-1},u_k))$. Then, by definition, 
$G_\alpha$ contains the $k$-clique $\{(u_1,1),\ldots,(u_k,k)\}$.
\end{proof}


\subsection{$\kSUM$ is in $\W[1]$} \label{main-reduction}

Using the above lemmas, we can efficiently reduce $\kSUM$ to $\kClique$. 
Consider a $\kSUM$ instance $(S,t)$ where $S=\{ x_1, \ldots, x_n \} \subseteq 
[0,M]$ and $t\in[0,kM]$ with $M = n^{2k}$. Let $G=(V,E)$ be a node-weighted 
clique on $n$ nodes $V=\{ v_1,\ldots, v_n \}$ with weight function 
$w:V\rightarrow [0,M]$ such that $w(v_i) = x_i$ for all $i \in [n]$.
Clearly, $(S,t)$ has a $\kSUM$ solution if and only if the instance $(G,w,t)$ of 
$\nwkcliquename{k}{n}{n^2}{M}$ has a solution.

Set $d = \lceil \log{n} /\log\log{n} \rceil$ and $p=\lceil \log^{4k}{n} \rceil$, so that $p^d \geq (n)^{4k} > kM$. Using 
Lemma~\ref{lem:ewclique} the instance $(G,w,t)$ of 
$\nwkcliquename{k}{n}{n^2}{M}$ can be reduced to 
$O(k^d)=O(n^{\log{k}/\log\log{n}})$ instances of 
$\ewkcliquename{k}{n}{n^2}{M'}$, where $M' = O(k^3 \cdot 
\log^{8k+1}{n}/\log\log n)$. Then, using Lemma~\ref{weight-removal}, we can generate 
$g(n,k) = O(n^{\frac{\log{k}}{\log\log{n}}} \cdot k^{3k^2} \log^{8k^2+k} n)$ 
graphs on $n$ nodes and $O(n^2)$ edges such that some graph has a $k$-Clique if 
and only if the original $\kSUM$ instance has a solution.

For constant $k$, note that $g(n,k) = n^{o(1)}$, and hence:

\begin{theorem}
For any $c > 2$, if $\kClique$ can be solved in time $O(n^c)$, then $\kSUM$ can 
be solved in time $n^{c+o(1)}$.
\end{theorem}

Furthermore, we remark that by applying the above reduction from $\kSUM$ to 
$\kClique$ to the respective {\em unparameterized} versions of these problems, 
we obtain a reduction from $\subsetsum$ on arbitrary weights to $\eewc$ with 
small edge weights.

\begin{corollary}
\label{cor:subsetsum}
For any $\eps>0$, $\subsetsum$ on $n$ numbers in $[-2^{O(n)},2^{O(n)}]$ can be 
reduced to $2^{\eps n}$ instances of $\eewc$ on $n$ nodes with edge weights are 
in $[-n^{O(1/\eps)},n^{O(1/\eps)}]$.
\end{corollary}

Note that $\subsetsum$ on $n$ numbers in $[-2^{O(n)},2^{O(n)}]$ is as hard as 
the general case of $\subsetsum$ (by Lemma~\ref{random-prime}), and the fastest 
known algorithm for $\subsetsum$ on $n$ numbers runs in time $O(2^{n/2})$.
The unweighted $\maxclique$ problem, which asks for the largest clique in a 
graph on $n$ nodes, can be solved in time $O(2^{n/4})$ \cite{R01}.
Corollary~\ref{cor:subsetsum} shows that even when the edge weights are small, 
the edge-weighted version of $\maxclique$ requires time $\Omega(2^{n/2})$ unless 
$\subsetsum$ can be solved faster.

\paragraph{An FPT Reduction.} We show how to make the reduction fixed-parameter 
tractable. We can modify the oracle reduction for $\kClique$ above to get a 
many-one reduction to $\kClique$ if we simply take the disjoint union of the 
$g(n,k)$ $\kClique$ instances as a single $\kClique$ instance. The resulting 
graph has $n \cdot g(n,k)$ nodes, $O(n^2 \cdot g(n,k))$ edges, and has a 
$k$-clique if and only if one of the original graphs has a $\kClique$. Then, we 
make the following standard argument to appropriately bound $g(n,k)$ via case 
analysis. If $k < \lceil \log \log n \rceil$, then $g(n,k) \leq n^{o(1)} \cdot 
2^{f(k) \cdot \poly(k)}$. If $ k \geq \lceil \log \log n \rceil$, then since $n 
\leq 2^{2^k}$, we have that $g(n,k) \leq 2^{2^k + f(k) \cdot \poly(k)}$. 
Therefore, $g(n,k) \leq n^{o(1)} \cdot h(k)$ for some computable $h : \N \to 
\N$, and we have shown the following:

\begin{lemma} \label{main-in-w1}
	$\kSUM$ is in $\W[1]$.
\end{lemma}

In Appendix~\ref{reducing},
we show how to obtain a randomized FPT reduction from the $\kSUM$ problem over 
the integers to $\kClique$, and how under plausible circuit lower bound 
assumptions, we can derandomize this reduction to show that $\kSUM$ over the 
integers is in $\W[1]$. This yields the first half of Theorem~\ref{thm:w1} (and 
we show the remainder, that $\kSUM$ is $\W[1]$-hard, in the next section).


\subsection{Node-Weight $\kClique$-Sum} \label{sec:nw}

The reduction of Section~\ref{main-reduction} shows that the Node-Weight 
$\kClique$-Sum problem can be reduced to $n^{o(1)}$ instances of $\kClique$, 
when $k$ is a fixed constant. We observe that if the input graph has $m$ edges, 
then the graphs generated by the reduction have no more than $k^2 m$ edges. 
Therefore, we have a tight reduction from node-weight clique to $\kClique$.

This concludes the proof of the first half of Theorem~\ref{thm:NW} referencing 
$\kClique$. We defer the proof of the second half of Theorem~\ref{thm:NW} 
concerning $\kDS$ to Appendix~\ref{kds-appendix}.


%

\section{From $\kClique$ to $\kSUM$} \label{sec:cliquetosum}

In this section, we give a new reduction from $k$-clique to $\kSUM$ in which the 
numbers generated are all in the interval $[-n^{2k}, n^{2k}]$. This proves that 
$\kSUM$ is in fact $\W[1]$-hard. We can view the result as an alternate proof 
for the $\W[1]$-hardness of $\kSUM$ without use of the Perfect Code problem, as 
done by Downey and Fellows~\cite{DF92}. The reduction is given from $\kClique$ 
to $\kVectorSUM$ (recall Definition~\ref{kvecsum-def}), and then from 
$\kVectorSUM$ to $\kSUM$.


\begin{lemma}\label{cliquetosum}
	For an integer $k>1$, $\kClique$ on $n$ nodes and $m$ edges reduces to an 
	instance of $\kvecsumname{k + \binom{k}{2}}{kn+{k \choose 2}m}{k \cdot n^{1 + 
	o(1)}}{k^2 + k + 1}$ deterministically in time $O(n^2)$.
\end{lemma}

\begin{proof}
	Given a graph $G = (V,E)$ we reduce an instance of $\kClique$ to an instance 
	of $\kVectorSUM$.
	First, we construct a $\ksumfree$ set $D \subseteq [Q]$ of size $n$, where $Q$ 
	can be bounded by $n^{1 + o(1)}$, and uniquely associate each vertex $v \in V$ 
	with an element $q_v \in D$. By Lemma~\ref{pr:kss}, the set $D$ can be 
	constructed in time $\poly(n)$.

	Let $\{ e_1, \ldots, e_d \}$ be the standard basis of $\mathbb{R}^d$. Let $T = 
	Q(k-1) + 1$, $d = k^2 + k + 1$, and define $\indvec{i}{j}$ to be the 
	length-$d$ vector ($1$-indexed) with entry $i$ consisting of the value $j$ and 
	all other entries being $0$.
	We define a mapping $\func : (V \times [k]) \cup (E \times [k] \times [k]) \to 
	\ints{0}{T}^d$ from the vertices and edges of the graph to length-$d$ vectors 
	over $\ints{0}{T}$.
	For each vertex $v \in V$ and $i \in [k]$, define
	\[ \func(v, i) \triangleq \indvec{i}{(T - (k-1) q_v)} + e_d \]
	and for each edge $e = (u,v) \in E$ and $i,j \in [k]$ with $i < j$, define
	\[
	\func(e,i,j) \triangleq \indvec{i}{q_u} + \indvec{j}{q_v} + e_{k \cdot i + j}.
	\]
	Define the target vector $\boldt = \indvec{d}{k} + (\sum_{i=1}^k 
	\indvec{i}{T}) + (\sum_{i,j\in[k], i<j} e_{k \cdot i + j})$. Then, there is a 
	$k$-clique in the instance $G$ of $\kClique$ on $n$ nodes and $m$ edges if and 
	only if there is a solution to the instance $I = ( \{\func(v,i) \mid v \in V,\ 
	i \in [k] \} \cup \{\func(e,i,j) \mid e \in E, \ i,j \in [k],\ i < j\}, 
	\boldt)$ of $\kvecsumname{k + \binom{k}{2}}{kn+\binom{k}{2}m}{T}{d}$.

	We now show correctness of the reduction. First, assume $S=\{ 
	u_1,\ldots,u_k\}\subseteq V$ forms a $k$-clique in $G$. Let $E(S)$ be the set 
	of edges between vertices in $S$.	We claim the set \[ A = \{\func(u_j,j) \mid 
	u_j \in S \} \cup \{\func(e,i,j) \mid e = (u_i, u_j) \in E(S), \ i<j\}, \] 
	with $|A| = k+\binom{k}{2}$, is a solution to the instance $I$.
	To see this, consider $\sum_{\boldv \in A} \boldv = \sum_{u_j \in S} 
	\func(u_j,j) + \sum_{e = (u_i,u_j) \in E(S)} \func(e,i,j)$ and note that it 
	equals \[ \sum_{j=1}^k \left( (T - (k-1) q_{v_j}) \cdot e_j +	e_d \right) + 
	\sum_{j=1}^k (k-1) \indvec{j}{q_{v_j}} + \sum_{i,j \in [k], i<j} e_{k \cdot i 
	+ j} = \boldt. \]
		
	For the other direction, consider some subset of vectors $A \subseteq  
	\{\func(v,i) \mid v \in V,\ i \in [k] \} \cup \{\func(e,i,j) \mid e \in E,\ 
	i,j \in [k],\ i \neq j\}$ of size $k+\binom{k}{2}$ for which the sum $\boldz = 
	\sum_{\boldv \in A} \boldv$ equals the target vector $\boldt$.
	Let $A_V \subseteq V \times [k]$ represent the set of pairs $(u,i)$ for $u \in 
	V$ and $i \in [k]$ such that $\func(u,i) \in A$. Similarly, let $A_E \subseteq 
	E \times [k] \times [k]$ represent the set of triples $(e,i,j)$ for $e \in E$ 
	and $i,j \in [k]$ such that $\func(e,i,j) \in A$.

	Recall that $\evec{\boldt}{d} = k$ and therefore by the definition of the 
	mapping $\func$ from vertices to vectors, $|A_V|$ must equal $k$. Since 
	$|A|=|A_V|+|A_E|$, it follows that $|A_E|=\binom{k}{2}$. Moreover, for all 
	$i,j \in [k]$ with $i<j$, since we defined $\evec{\boldt}{k\cdot i + j}=1$, 
	there is exactly one triple in $A_E$ of the form $(e,i,j)$.
	This implies that for every $j\in[k]$, there is at least one pair $(u,j) \in 
	A_V$, for otherwise $\evec{\boldz}{j}$ would be the sum of only $(k-1)$ 
	numbers from $D \subseteq [Q]$, and since $\evec{\boldt}{j} = T > Q (k-1)$, 
	this sum cannot equal $\evec{\boldt}{j}$. Note that $|A_V|=k$, and so for each 
	$j \in [k]$ there is exactly one pair $(u,j) \in A_V$. Denote this vertex $u$ 
	by $u_j$. 


	For each $j\in[k]$, let $e_1,\ldots,e_{k-1}$ be the $k-1$ edges in $E$ such 
	that for each $i \in [k-1]$, either $(e_i,y,j) \in A_E$ or $(e_i,j,y)\in A_E$ 
	for some $y \in [k]$. If the former holds, then let $e_i = (x,v_i)$ and if the 
	latter holds, let $e_i = (v_i,x)$ for some $x, v_i \in V$. We claim that for 
	$i \in [k-1]$, $v_i$ must be identical to $u_j$.
	To see this, note that by the definition of the vectors, for all $j \in [k]$, 
	$\evec{\boldz}{j} = \evec{\func(u_j,j)}{j} + \sum_{i=1}^{k-1} 
	\evec{\func(e_{i},i,j)}{j} = T - (k-1) q_{u_j} + \sum_{i=1}^{k-1} q_{v_i}$, 
	and since $\evec{\boldz}{j} = \evec{\boldt}{j}=T$, it must be the case that 
	$\sum_{i=1}^{k-1} q_{v_i} = (k-1) q_{u_j}$. Then, since $D$ is $\ksumfree$, 
	this implies that $u_j = v_i$ for all $i\in[k-1]$, which concludes the claim.

	Therefore, the vectors in $A_V$ correspond to $k$ vertices $V(A)\subseteq V$ 
	and the vectors in  $A_E$ correspond to $\binom{k}{2}$
	edges $E(A)\subseteq E$ such that precisely $k-1$ edges are incident to each 
	vertex in $V(A)$. Thus, each edge must be incident to two vertices in $V(A)$, 
	and so the vertices in $V(A)$ along with the edges in $E(A)$ form a $k$-clique 
	in $G$.
\end{proof}

The following lemma gives a simple reduction from $\kVectorSUM$ to $\kSUM$, by 
the usual trick of converting from vectors to integers (via a Freiman 
isomorphism of order $k$).

\begin{lemma} \label{vect-sum}
	$\kvecsumname{k}{n}{M}{d}$ can be reduced to $\kSUM$ on $n$ integers in the 
	range $[0,(kM+1)^d]$ in $O(n\log{M})$ time.
\end{lemma}

\begin{proof}
Let $p = kM+1$. Let $\boldv_1, \ldots, \boldv_n$ be the $n$ vectors of a 
$\kVectorSUM$ instance with target $\boldt$. For each vector $\boldv_i = \langle 
\vid{i}{j} \rangle_{j=0}^{d-1}$, define the integer $x_i \triangleq 
\sum_{j=0}^{d-1} \vid{i}{j} \cdot p^j$, and define $t = \sum_{j=0}^{d-1} 
\boldt_j \cdot p^j$. Then the $\kSUM$ instance $(\{ x_i \}_{i \in [n]}, t)$ has 
a solution if and only if the $\kVectorSUM$ instance $(\{ \boldv_i \}_{i \in 
[n]}, \boldt)$ has a solution. For correctness, let $S = \{i_1, \ldots, i_k \}$ 
be a set of $k$ indices. Then, $\sum_{i \in S} x_i = \sum_{i \in S} 
\sum_{j=0}^{\ell-1} \vid{i}{j} \cdot p^j$. By switching the order of the 
summations, we have that this quantity is $\sum_{j=0}^{\ell-1} \sum_{i \in S} 
\vid{i}{j} \cdot p^j$. Since $\sum_{i \in S} \vid{i}{j} \in \ints{0}{p-1}$, we 
conclude by the uniqueness of the representation of $t$ that $\boldt_j = \sum_{i 
\in S} \vid{i}{j}$ for all integers $j \in \ints{0}{d-1}$. Hence, the 
$\kVectorSUM$ instance has a solution  described by the index set $S$ as well. 
For the other direction, if $\boldt_j = \sum_{i \in S} \vid{i}{j}$ for all $j 
\in \ints{0}{d-1}$, then $\sum_{j=0}^{d-1} \boldt_j \cdot p^j = \sum_{j=0}^{d-1} 
\sum_{i \in S} \vid{i}{j} \cdot p^j$. The left side is equal to $t$ and, by a 
switch of the order of summations, the right side is equal to $\sum_{i \in S} 
x_i$. Therefore, $t = \sum_{i \in S} x_i$ and so the $\kSUM$ instance has a 
solution described by the index set $S$, as desired.
\end{proof}

We remark that in some cases, the proof can be changed slightly to yield smaller 
numbers in the $\kSUM$ instance produced by the reduction. In particular, when 
reducing $\kClique$ to $\kVectorSUM$, only the numbers in the first $k$ 
coordinates can be as large as $k\cdot n^{1+o(1)}$ while the numbers in the last 
$k^2+1$ coordinates are bounded by $k$, and therefore, when reducing to 
$\ksumname{k+{k \choose 2}}{n}{M}$ on $kn+{k\choose 2} m$ numbers, the numbers 
generated can be bounded by $M=k^d\cdot (kn^{1+o(1)})^k \cdot k^{k^2+1} = 
O(k^{2k^2} \cdot n^{k+o(k)})$. In other words, we have reduced $\kClique$ to 
$k'\text{-SUM}$ with numbers in the range $\left[-n^{\sqrt{k'}}, n^{\sqrt{k'}} 
\right]$, where $k' = k + \binom{k}{2}$.

\medskip

The composition of Lemma~\ref{cliquetosum} and Lemma~\ref{vect-sum} yields an 
FPT reduction, and we have obtained:

\begin{lemma} \label{main-w1-hard}
	$\kSUM$ is $\W[1]$-hard.
\end{lemma}

This concludes the proof of Theorem~\ref{thm:w1}.

\paragraph{Acknowledgements.} We would like to thank the anonymous reviewers for 
their helpful comments. This work was supported in part by a David Morgenthaler 
II Faculty Fellowship, and NSF CCF-1212372.

\bibliographystyle{newbib2}
\bibliography{ksum}

\appendix

\section{Additional Preliminaries} \label{moreprelims}

We review some useful combinatorial results that are needed in 
Section~\ref{sec:cliquetosum}.

\begin{definition}[$(n,k)$-perfect hash function family]
	For integers $n,k,s > 0$, a set of functions $F = \{ f_i \}_{i \in [s]}$ where 
	$f_i : [n] \to [k]$ is a {\em $(n,k)$-perfect hash function family of size 
	$s$} if and only if for every size-$k$ subset $S \subseteq [n]$, there exists 
	some $f_i \in F$ such that $\cup_{j \in S} f_i(j) = [k]$.
\end{definition}

\begin{proposition}[{cf.~\cite[Theorem 3]{NSS95}}] \label{pr:PHF}
	An $(n,k)$-perfect hash function family of size $e^k \cdot k^{O(1)} \log n$ 
	can be constructed deterministically in time $k^{O(k)} \cdot \poly(n)$.
\end{proposition}

\begin{lemma}
	\label{pr:kss}
	For any $\eps > 0$, there exists a $c > 0$ such that a $\ksumfree$ set $S$ of 
	size $n$ with $S \subset [0, (ckn)^{1 + \eps}]$ can be constructed in 
	$\poly(n)$ time.
\end{lemma}

\begin{definition}[$\ksumfree$ set]
	For any $k \geq 2$, a set $S \subseteq \Z$ is a {\em $\ksumfree$ set} if and 
	only if for all (not necessarily distinct) $x_1, \ldots, x_k \in S$, $\sum_{i 
	\in [k-1]} x_i = (k-1) \cdot x_k$ implies that $x_1 = \cdots = x_k$.
\end{definition}

We outline a construction of a $\ksumfree$ set, derived from a small 
modification to Behrend's original construction~\cite{Beh46}.

\paragraph{Constructing $\ksumfree$ Sets.} For a vector $\boldv$, let $\| \boldv 
\|$ represent the $\ell_2$ norm of $\boldv$, and let $| \langle \boldv_1, 
\boldv_2 \rangle |$ represent the dot product of two vectors $\boldv_1$ and 
$\boldv_2$. Let $p = (k-1)d-1$. Consider the set $S_r(m,d)$ consisting of all 
integers $x$ that can be written as $x = \sum_{i=0}^{m-1} a_i p^{i}$, where for 
all $i \in [0,m-1]$, $a_i \in [0,d-1]$ and $\| \langle a_0, \ldots, a_{m-1} 
\rangle \| = \sqrt{r}$. For an integer $x \in S_r(m,d)$ of the form $x = \sum_{i 
\in [0,m-1]} a_i \cdot p^i$, let $\| x \|$ represent the $\ell_2$ norm of the 
vector $\langle a_i \rangle_{i \in [0,m-1]}$. By definition, $\| x \| = 
\sqrt{r}$.

\begin{claim} \label{cauchy}
	For any vectors $\boldv_1, \ldots, \boldv_k$, if $\| \boldv_1 \| = \cdots = \| 
	\boldv_k \|$ and $\sum_{i \in [k]} \| \boldv_i \| = \left\| \sum_{i \in [k]} 
	\boldv_i \right\|$, then $\boldv_1 = \cdots = \boldv_k$.
\end{claim}

\begin{proof}
	Note that
	\[ \left\| \sum_{i \in [k]} \boldv_i \right\|^2 = \sum_{i \in [k]} \langle 
	\boldv_i, \boldv_i \rangle  + 2 \sum_{i, j \in [k], i < j} | \langle \boldv_i, 
	\boldv_j \rangle | = \sum_{i \in [k]} \| \boldv_i \|^2 + 2 \sum_{i, j \in [k], 
	i < j} | \langle \boldv_i, \boldv_j \rangle |. \]

	By the Cauchy-Schwarz inequality, $| \langle \boldv_i, \boldv_j \rangle | \leq 
	\| \boldv_i \| \cdot \| \boldv_j \|$. Thus, in order for the above quantity to 
	be equal to $ \left(  \sum_{i \in [k]} \left\| \boldv_i \right\| \right)^2$, 
	it must be the case that $| \langle \boldv_i, \boldv_j \rangle | = \| \boldv_i 
	\| \cdot \| \boldv_j \|$ for all $i,j\in[k]$. Hence, also by Cauchy-Schwarz, 
	$\boldv_i$ and $\boldv_j$ are linearly dependent. Then, using the fact that 
	$\| \boldv_i \| = \| \boldv_j \|$, it follows that $\boldv_i = \boldv_j$ for 
	all pairs $i,j \in [k]$, which completes the proof.
\end{proof}

\begin{lemma}
	For all integers $k \geq 2$, for all $\ell \in [0, m(d-1)^2]$, $S_\ell$ is 
	$\ksumfree$.
\end{lemma}

\begin{proof}
	Let $x_1, \ldots, x_k \in S_\ell(m,d)$ be such that $\sum_{i \in [k-1]} x_i = 
	(k-1) \cdot x_k$. Then, $\| \sum_{i \in [k-1]} x_i \| = (k-1) \cdot \| x_k 
	\|$. Since $\| x_1 \| = \| x_2 \| = \cdots = \| x_k \|$, it follows that
	\[ \sum_{i \in [k-1]} \| x_i \| = \left\| \sum_{i \in [k-1]} x_i \right\|. \]
	Hence, by Claim~\ref{cauchy}, it follows that $x_1 = \cdots = x_k$.
\end{proof}

There are $d^m$ possible settings for $a_0, \ldots, a_{m-1}$ where each $a_i \in 
[0,d-1]$, and the value $\| \langle a_i \rangle_{i \in [0,m-1]} \|^2$ lies 
within the range $[0, m(d-1)^2]$. Hence, by the pigeonhole principle, there 
exists some $r \in [0, m(d-1)^2]$ such that $|S_r(m,d)| \geq \frac{d^m}{ 
m(d-1)^2 + 1} \geq d^{m-2} / m$. Note that $S_r(m,d) \subset [0, p^m]$. Hence, 
to obtain a $\ksumfree$ set of size $n$, we have that $d = (m \cdot 
n)^{\frac{1}{m-2}}$. Recall that $p = \Theta(k d)$, and so $p^m = O(k m 
n)^{\frac{m}{m-2}}$. For any $\eps > 0$, we can set $m = 2/\eps + 2$ so that 
$p^m = (ckn)^{1 + \eps}$ for some constant $c > 0$. Hence, $S_r(m,d) \subset [0, 
(ckn)^{1 + \eps}]$. For any $n$, we can compute, for each $\ell \in [0, m 
(d-1)^2]$, the exact size of $S_\ell(m,d)$, returning the $\ell$ for which 
$|S_\ell(m,d)|$ is maximized. Note that $m (d-1)^2 = O(n^\eps)$, and the 
construction of $S_\ell(m,d)$ requires time linear in $|S_\ell(m,d)|$. 
Lemma~\ref{pr:kss} follows.

\subsection{Background on Parameterized Complexity}

For an exposition of fixed  parameter tractability, reducibility, and the 
$\W$-hierarchy, we refer the reader to~\cite{DF92,DF95,FG06}. Here we give a 
brief overview.

A parameterized problem is defined to be a subset of $\{0,1\}^{\star} \times 
\N$. Let $L, L'$ be parameterized problems. $L$ is in the class $\FPT$ if there 
is an algorithm $A$, constant $c$, and computable function $f : \N \rightarrow 
\N$, such that on all inputs $y = (x,k)$, $A(y)$ decides whether $y$ is in $L$ 
and runs in time at most $f(k)\cdot |x|^c$.  A \emph{FPT reduction} from $L$ to 
$L'$ is an algorithm $R$ mapping $\{0,1\}^{\star} \times \N$ to $\{0,1\}^{\star} 
\times \N$, such that for all $y=(x,k)$, $R(y) \in L'$ if and only if $y \in L$, 
$R(y)$ runs in $f(k)\cdot |x|^c$ time (for some $c$ and $f$), and $R(y)=(x',k')$ 
where $k' \leq g(k)$ for some computable function $g$.

Although it is not the original definition, it is equivalent to define $\W[1]$ 
to be the class of parameterized problems which have an FPT reduction to 
$k$-Clique~\cite{DF95,FG06}. 

\section{A Deterministic FPT Reduction from $\kSUM$ to $\kClique$} 
\label{reducing}

In this section we show how, by assuming a plausible circuit lower bound 
assumption, we are able to strengthen our result from 
Section~\ref{main-reduction} to a reduction from general $\kSUM$ to $\kClique$ 
(as opposed to only reducing $\kSUM$ on numbers in $[-n^{2k},n^{2k}]$ to 
$\kClique$).

To do this, we give a randomized process which takes a $\kSUM$ instance and 
outputs a collection of $\kSUM$ instances on numbers in $[-n^{2k},n^{2k}]$ 
which, with high probability, are such that a member of the collection contains 
a solution if and only if the original $\kSUM$ instance contains a solution. 
Then, we apply a theorem which allows us to derandomize this reduction given the 
appropriate circuit lower bound.

\paragraph{A Randomized Weight Reduction.} We first define a randomized process 
(cf.~\cite[Definition~4.1]{KM02}) consisting of an algorithm $F$ along with a 
predicate $\pi$. In what follows, the algorithm $F$ will represent a 
(randomized) oracle reduction from general $\kSUM$ to $\kSUM$ on numbers in 
$[-n^{2k},n^{2k}]$ which succeeds on some random bit sequences, and the 
predicate $\pi$ represents whether or not the input random bit sequence is such 
that $F$ produces a valid oracle reduction. Let $z = O(k \cdot \poly \log n)$, 
and let $\calD$ represent the domain of $\kSUM$ instances on $n$ integers.

Let $F : \calD \times \{0,1\}^z \to \calD^k$ be a process which takes as input 
$x$ a $\kSUM$ instance along with a random bit sequence $r \in \{0,1\}^z$. Let 
the instance $x$ describe a $\kSUM$ instance consisting of a set $S$ of $n$ 
integers in the range $[0,M]$ with target value $0$ (without loss of 
generality). $F$ uses the sequence of random bits $r$ to choose a prime $p \in 
[2, d n^k \log n \log(kM)]$ by picking random integers within the interval until 
a prime is obtained. Let $S'$ represent the set of integers derived by taking 
each integer in $S$ modulo $p$. The output of $F$ is a collection of $k$ 
instances of $\kSUM$, each with base set $S'$, and the $i^\th$ instance has 
target value $i \cdot p$, for  $i \in [0,k-1]$.

Let $\pi : \calD \times \{0,1\}^z \to \{0,1\}$ be a predicate which takes as 
input $x$ a $\kSUM$ instance along with a random bit sequence $r \in \{0,1\}^z$ 
and outputs a single bit. The predicate $\pi(x,r)$ runs $F(x,r)$ to obtain a 
collection $C$ of $k$ instances of $\kSUM$, and outputs $1$ if and only if the 
$\kSUM$ instance $x$ has a solution and there exists some $\kSUM$ instance in 
$C$ which does not have a solution, or, the original $\kSUM$ instance does not 
have a solution and none of the $\kSUM$ instances in $C$ have a solution.
The following lemma is folklore, and also informally stated in~\cite{BDP08}, but 
we provide a proof here for completeness, as it also justifies 
Proposition~\ref{random-prime-prop}.

\begin{lemma} \label{random-prime}
	Let $k,n,M,M' > 0$ be integers such that $M' = O( k n^k \log n \log(kM))$, and 
	let $d > 0$ be a sufficiently large constant. Let $C$ be the collection of 
	$\kSUM$ instances output by $F$ on input $x$, a $\kSUM$ instance on $n$ 
	integers in the range $[0,M]$, and a random bit sequence $r \in \{0,1\}^z$. 
	Then, each $\kSUM$ instance in $C$ is on $n$ integers in the range $[0,M']$, 
	and $\pi(x,r) = 1$ with probability at least $1 - 1/(cd)$. for some fixed 
	constant $c > 0$.
\end{lemma}

\begin{proof}
	Let $S$ be the set of integers of the $\kSUM$ instance. For a sufficiently 
	large constant $d > 0$, for a randomly chosen prime $p \in [2,d \cdot n^k \log 
	n \log(kM)]$, consider the new set $S'$ where we take all integers of $S$ 
	modulo $p$. Clearly, if $S$ has a $\kSUM$ solution, then $S'$ also has a 
	solution, and hence, there exists a member of the collection of $\kSUM$ 
	instances output by $F$ which contains a solution. For any $k$-tuple of 
	integers in $[0,M]$ with nonzero sum $s$, we have $s \in [kM]$ so $s$ has at 
	most $\log(kM)$ prime factors. By the prime number theorem, for sufficiently 
	large $n$, there are at least $c d \cdot n^k \log(kM)$ primes in this interval 
	for a fixed constant $c > 0$, and so the probability that $p$ is a prime 
	factor of $s$ is at most \[\frac{\log(kM)}{c d n^k \log(kM)} \leq 
	\left(\frac{1}{c d \cdot n^k}\right).\] By a union bound over all $k$-tuples 
	of numbers from $S'$, the probability of a false positive in $S'$ is at most 
	$1/(cd)$. We have shown that if the $\kSUM$ instance $S$ with target value $0$ 
	does not contain a solution, then with probability at least $1 - 1/(cd)$ over 
	the choice of a random prime $p$, $S$ does not contain a $k$-tuple of integers 
	which sum to a multiple of $p$, and hence, there does not exist a member of 
	the collection of $\kSUM$ instances output by $F$ which contains a solution.

	Finally, we can remove the ``modulo $p$'' constraint from $S'$ as follows.
	Treating all elements in $S'$ as integers in $[0,p-1]$, for each $z \in 
	\{-k\cdot p,\ldots,-p,0,p,\ldots,k\cdot p\}$ we reduce to checking whether 
	there are $x_1,\ldots,x_k \in S'$ such that $\sum_i x_i = z$ (where this sum 
	is over the integers). Such a $k$-tuple exists if and only if $\sum_i x_i$ 
	equals $0 \text{ mod } p$.
\end{proof}

\paragraph{Derandomizing the Reduction.} Let $B$ be an oracle for $\SAT$, and 
assume there exists a function $f \in \E$ with $B$-oracle circuit complexity 
$2^{\eps n}$ for some $\eps > 0$. We note that, for all $\kSUM$ instances $x$ 
and random bit sequences $r$, $\pi(x,r)$ can be decided by a $B$-oracle circuit 
of size $n^a$ for some positive constant $a$, since we can construct a circuit 
which creates and then solves the $\kSUM$ instances to check the predicate. We 
also have by Lemma~\ref{random-prime} that for all $x$, for uniformly sampled $r 
\in \{0,1\}^z$, $\pr[\pi(x,r) = 1]$ with probability $1 - 1/(cd)$, where $c$ is 
a universal constant and $d$ can be chosen to be arbitrarily large. We now apply 
the following derandomization theorem due to Klivans and van 
Melkebeek~\cite{KM02}.

\begin{theorem}[{cf.~\cite[Theorem 4.4]{KM02}}] \label{derand}
	Let $B$ be an oracle, $b$ a positive constant, and $\ell : \N \to \N$ a 
	constructible function. Let $(F, \pi)$ be a randomized process using a 
	polynomial number of random bits such that $B$ can efficiently check $(F, 
	\pi)$. If there exists a Boolean function $f \in \E$ such that $C_f^B(\ell(n)) 
	= \Omega(n)$, then there exists a function $G$ computable in $\E$ and a 
	constructible function $s(n) = O(\ell^2(n^{O(1)}) / \log n)$ such that for any 
	input $x$ of length $n$,
	\[ \left| \pr_r[ \pi(x,r) = 1] - \pr_s[ \pi(x,G(s)) = 1] \right| \leq O(1 / 
	n^b) \]
	where $r$ is uniformly distributed over $\{0,1\}^{r(n)}$ and $s$ over 
	$\{0,1\}^{s(n)}$.
\end{theorem}

By applying Theorem~\ref{derand} with $\ell(n) = \log n$, there exists a 
constant $s > 0$ and an efficiently computable function $G$ computable in with 
seed length $s \log n$ such that for any input $x$ of length $n$, $| 
\pr[\pi(x,r) = 1] - \pr[\pi(x,G(\sigma)) = 1]| = O(1/n^b)$. By choosing $d = 
\Omega(n)$, we have that $\pr[\pi(x,G(\sigma)) = 1] \geq 1 - 1/n$.

\paragraph{The Deterministic FPT Reduction.} Given an instance $x$ of $\kSUM$, 
for each $\sigma \in \{0,1\}^{s \log n}$, let $C_\sigma$ be the collection of 
instances of $\kSUM$ output by $F(x,G(\sigma))$. Note that $C_\sigma$ is in fact 
a collection of instances of $\kSUM$ on numbers in $[-n^{2k},n^{2k}]$ by 
Lemma~\ref{random-prime}. Hence, as described in Section~\ref{main-reduction}, 
we can reduce each member of $C_\sigma$ to an instance of $\kClique$---let 
$S_\sigma$ represent this collection of $\kClique$ instances, obtained by 
applying the reduction to each member of $C_\sigma$. Let $S^*$ be a family of 
$\kClique$ instances obtained by taking the union of $S_\sigma$ over all $\sigma 
\in \{0,1\}^{s \log n}$. The original $\kSUM$ instance $x$ contains a solution 
if and only if at least a $1/n$ fraction of the members of $S^*$ contain a 
$\kClique$.

To see the correctness of this deterministic reduction, note that at most a 
$1/n$ fraction of $\sigma \in \{0,1\}^{s \log n}$ are such that 
$\pi(x,G(\sigma)) = 0$. For all other seeds $\sigma$, $\pi(x,G(\sigma))=1$ and 
the process $F$ correctly reduces its input $x$ to a collection of $\kSUM$ 
instances. Therefore, if $x$ contains a solution, then all members of $S^*$ will 
contain a $\kClique$, and if $x$ does not contain a solution, then at most a 
$1/n$ fraction of the members of $S^*$ could possibly contain a $\kClique$. As 
the constant $s$ in the seed length is fixed, we note that the process $F$ also 
runs in time polynomial in its input and therefore our reduction is FPT. We have 
shown the following theorem.

\begin{theorem} \label{thm:derand1}
	Let $B$ be an oracle for $\SAT$. If there is a function $f \in \E$ with 
	$B$-oracle circuit complexity $2^{\eps n}$ for some $\eps > 0$, then $\kSUM 
	\in \W[1]$.
\end{theorem}

We note that the oracle $B$ can be replaced with any oracle sufficient to 
efficiently compute the predicate $\pi$---for example, Theorem~\ref{thm:derand1} 
follows where $B$ is an oracle for $\kSUM$.

\section{Node-Weight Dominating Set} \label{kds-appendix}

\begin{lemma}
For every fixed $k\geq 2$, Node-Weight $k$-Dominating-Set-Sum can be reduced to 
$n^{o(1)}$ instances of $\kDS$ on graphs on $O(k^2 n)$ nodes. The reduction runs 
in $n^{2+o(1)}$ time.
\end{lemma}

\begin{proof}

Given a node-weighted graph $G=(V,E), w:V \rightarrow [ -n^{2k}, n^{2k} ]$, we 
first use Lemma~\ref{lem1} with $d=O(\log{n}/\log\log{n})$ and 
$p=O(k2^k\log{n})$ to get $s=O(n^{\log{k}/\log\log{n}})=n^{o(1)}$ instances of 
the problem where the weight of every node is a vector in $[p]^d$ instead of a 
number in $[-n^{2k},n^{2k}]$. We treat each instance $i \in [s]$ separately, as 
follows.

For a pair of nodes $u,v \in V$, let their weights be $\boldu,\boldv \in [p]^d$ and define the expression \[F(u,v)= \sum\limits_{j=1}^d \left( \evec{\boldu}{j}^2 + \evec{\boldv}{j}^2 + 2(k-1) \evec{\boldu}{j} \cdot \evec{\boldv}{j} \right),\] which corresponds to the ``squaring trick" from Lemma~\ref{lem:ewclique}.

We enumerate over all ${k \choose 2}$-tuples of numbers $t=(\alpha_{i,j})_{i,j \in [k]}$ such that $\sum_{i,j} \alpha_{i,j} = 0$ where $\alpha_{i,j} \in [-M,M]$ where $M=O(\poly\log{n})$ is an upper bound on the $F(u,v)$ values, and for each such tuple $t$ we generate an instance of $\kDS$, an unweighted graph $G_t$. Note that the number of such tuples is $\log^{O(k)}{n}=n^{o(1)}$.

Let $V_0$, $V_1,\ldots,V_k$, and $V_{i,j}$ for every $i<j, i,j\in[k]$ be ${k \choose 2} + k+ 1$ copies of the node set $V$.
The node set of $G_t$ is $V_0 \cup \bigcup_i V_i \cup \bigcup_{i<j} V_{i,j}$. Let us denote the copy of node $v$ in the set $V_X$ by $v_X$. 
For every edge $(u,v) \in E$ of $G$ we add the edges $(u_i,v_0)$ and $(v_i,u_0)$, for every $i \in [k]$, to $G_t$. 
We also make the copies $V_1,\ldots,V_k$ cliques by adding edges $(v_i,u_i)$ for every $u,v \in V$ and $i\in[k]$. Note that this forces every $k$-dominating set in $G_t$ to contain exactly one node $v_i$ in every set $V_i$, and that such $k$ nodes dominate every node in $V_0$ if and only if they correspond to a $k$-dominating set in $G$.
Finally, we add the following edges that simulate ``gaining edges" by allowing pairs of nodes $u_i,v_j$ to be in a $k$-dominating set in $G_t$ if and only if they are consistent with the tuple $t$ which is chosen to guarantee that the total sum of weights is zero.  
For every pair of nodes $u \neq v$ in $G$ and indices $i<j$ in $[k]$ we add the edge $(u_i,v_{i,j})$ to $G_t$ and then check if $F(u,v)=\alpha_{i,j}$, and if so, we add the edge $(u_j,v_{i,j})$ to $G_t$.

We now claim that a subset $S \subseteq V$ of size $k$ is $k$-dominating set of weight zero in $G$ if and only if the nodes $S'=\{x_{i(x)} \mid x \in S \}$ where $i(x)$ is the index of $x$ in an arbitrary ordering of $S$, are a $k$-dominating set in $G_t$, for some tuple $t$ and some instance $i \in [s]$.
To see this, first note that a set of $k$ nodes $S[1]_1 \in V_1,\ldots,S[k]_k \in V_k$ are a dominating set in $G_t$ if and only if $S[1],\ldots,S[k]$ are a dominating set in $G$ and for every pair $i<j$ in $[k]$, the condition $F(w(S[i]),w(S[j]))=\alpha_{i,j}$ is satisfied.
By definition of our formulas $F(\cdot,\cdot)$ and the proofs of 
Lemmas~\ref{lem1}, \ref{lem:ewclique}, and \ref{weight-removal}, we get that 
$S[1]_1 ,\ldots,S[k]_k$ are a dominating set in $G_t$ for some tuple $t$ and 
instance $i\in[s]$ if and only if $S[1],\ldots,S[k]$ are a dominating set in $G$ 
of total weight zero.
\iflncs \qed \fi
\end{proof}

\end{document}